\documentclass[11pt]{article}

\usepackage[letterpaper,margin=1.00in]{geometry}
\usepackage{amsmath, amssymb, amsthm, thmtools, amsfonts}
\usepackage{bbm}

\usepackage{ifthen}

\usepackage{tikz}
\usetikzlibrary{positioning,decorations.pathreplacing}

\usepackage{cite}
\usepackage{appendix}
\usepackage{graphicx}
\usepackage{color}
\usepackage{algorithm}
\usepackage[noend]{algpseudocode}
\usepackage{epstopdf}
\usepackage[textsize=tiny]{todonotes}

\usepackage{framed}
\usepackage[framemethod=tikz]{mdframed}
\usepackage[bottom]{footmisc}
\usepackage[shortlabels]{enumitem}
\setitemize{noitemsep,topsep=3pt,parsep=3pt,partopsep=3pt}
\usepackage[font=small]{caption}
\usepackage{xspace}

\usepackage{mathtools}

\newtheorem{theorem}{Theorem}[section]
\newtheorem{lemma}[theorem]{Lemma}
\newtheorem{meta-theorem}[theorem]{Meta-Theorem}

\newtheorem{remark}[theorem]{Remark}
\newtheorem{corollary}[theorem]{Corollary}

\newtheorem{definition}[theorem]{Definition}

\definecolor{darkgreen}{rgb}{0,0.5,0}
\usepackage{hyperref}
\hypersetup{
    unicode=false,          
    colorlinks=true,        
    linkcolor=red,          
    citecolor=darkgreen,        
    filecolor=magenta,      
    urlcolor=cyan           
}
\usepackage[capitalize, nameinlink]{cleveref}
\crefname{theorem}{Theorem}{Theorems}
\Crefname{lemma}{Lemma}{Lemmas}
\Crefname{observation}{Observation}{Observations}
\Crefname{remark}{Remark}{Remarks}
\Crefname{equation}{}{}

\newcommand{\eps}{\varepsilon}

\newcommand{\local}{$\mathsf{LOCAL}$\xspace}

\newcommand{\poly}{\operatorname{\text{{\rm poly}}}}
\newcommand{\floor}[1]{\lfloor #1 \rfloor}
\newcommand{\ceil}[1]{\lceil #1 \rceil}

\newcommand{\set}[1]{\left\{#1\right\}}

\newcommand{\paren}[1]{\mathopen{}\left(#1\right)\mathclose{}}
\newcommand{\card}[1]{\left|#1\right|}

\renewcommand{\paragraph}[1]{\vspace{0.15cm}\noindent {\bf #1}:}

\usepackage{mathtools}

\begin{document}

\date{}

\title{Improved Deterministic Distributed Matching via Rounding}

\author{
	 Manuela Fischer\\
  \small ETH Zurich \\
  \small manuela.fischer@inf.ethz.ch
 }

\maketitle

\setcounter{page}{0}
\thispagestyle{empty}

\begin{abstract}
We present improved deterministic distributed algorithms for a number of well-studied matching problems, which are simpler, faster, more accurate, and/or more general than their known counterparts. The common denominator of these results is a \emph{deterministic distributed rounding} method for certain \emph{linear programs}, which is the first such rounding method, to our knowledge. 
A sampling of our end results is as follows.

\begin{itemize}
\item An $O\paren{\log^2  \Delta\cdot \log n}$-round deterministic distributed algorithm for computing a maximal matching, in $n$-node graphs with maximum degree $\Delta$. This is the first improvement in about 20 years over the celebrated $O(\log^4 n)$-round algorithm of Ha\'n\'ckowiak, Karo\'nski, and Panconesi [SODA'98, PODC'99]. 

\item A deterministic distributed algorithm for computing a $(2+\eps)$-approximation of maximum matching in $O\paren{\log^2 \Delta \cdot  \log \frac{1}{\eps} + \log^ * n}$ rounds. This is exponentially faster than the classic $O(\Delta +\log^* n)$-round $2$-approximation of Panconesi and Rizzi [DIST'01]. With some modifications, the algorithm can also find an $\eps$-maximal matching which leaves only an $\eps$-fraction of the edges on unmatched nodes.

\item An $O\paren{\log^2 \Delta \cdot \log \frac{1}{\eps} + \log^ * n}$-round 
deterministic distributed algorithm for computing a $(2+\eps)$-approximation of a maximum weighted matching, and also for the more general problem of maximum weighted $b$-matching. These improve over the $O\paren{\log^4 n \cdot \log_{1+\eps} W}$-round $(6+\eps)$-approximation algorithm of Panconesi and Sozio [DIST'10], where $W$ denotes the maximum normalized weight.

\item A deterministic \emph{Local Computation Algorithm} (LCA) for a $(2+\eps)$-approximation of maximum matching with $2^{O(\log^2{\Delta})}\cdot  \log^* n$ queries. This improves almost exponentially over the previous deterministic constant approximations with query-complexity $2^{\Omega(\Delta\cdot \log \Delta)}\cdot  \log^* n$. 
\end{itemize}

\end{abstract}

\newpage

\section{Introduction and Related Work}
\vspace{-6pt}

We work with the standard \local model of distributed computing \cite{linial1987LOCAL}: the network is abstracted as a graph $G=(V, E)$, with $n=|V|$, $m=|E|$, and maximum degree $\Delta$. Each node has a unique identifier. In each round, each node can send a message to each of its neighbors. We do not limit the message sizes, but for all the algorithms that we present, $O(\log n)$-bit messages suffice. We assume that all nodes have knowledge of $\log \Delta$ up to a constant factor. If this is not the case, it is enough to try exponentially increasing estimates for $\log \Delta$. 
\vspace{-8pt}
\subsection{Broader Context, and Deterministic Distributed Rounding}
\vspace{-6pt}
Efficient deterministic distributed graph algorithms remain somewhat of a rarity, despite the intensive study of the area since the 1980's. In fact, among the four classic problems of the area --- maximal independent set, $(\Delta+1)$-vertex-coloring, maximal matching, and $(2\Delta-1)$-edge-coloring --- only for maximal matching a $\poly\log n$-round deterministic algorithm is known, due to a breakthrough of Ha\'n\'ckowiak, Karo\'nski, and Panconesi \cite{HanckowiakKP98, hanckowiak1999faster}. Finding $\poly\log n$-round deterministic algorithms for the other three problems remains a long-standing open question, since \cite{linial1987LOCAL}. In a stark contrast, in the world of randomized algorithms, all these problems have $O(\log n)$-round \cite{luby1986simple, alon1986fast} or even more efficient algorithms \cite{barenboim2012locality, Ghaffari-MIS, harris2016distributed}. 

Despite this rather bleak state of the art for deterministic algorithms, there is immense motivation for them. Here are three sample reasons: (1) One traditional motivation is rooted in the classic complexity-theoretic quest which seeks to understand the difference between the power of randomized and distributed algorithms. (2) Another traditional motivation comes from practical settings where even small error probabilities cannot be tolerated. (3) Nowadays, there is also a more modern motive: we now understand that in order to have faster randomized algorithms, we \emph{must} come up with faster deterministic algorithms.\footnote{For instance, our improvement in the deterministic complexity of maximal matching directly improves the randomized complexity of maximal matching, as we formally state in \Cref{Crl:RandMM}.} This connection goes in two directions: (A) Almost all the recent developments in randomized algorithms use the \emph{shattering technique} \cite{barenboim2012locality, Ghaffari-MIS, harris2016distributed, GS17} which randomly breaks down the graph into small components, typically of size $\poly \log n$, and then solves them via a deterministic algorithm. Speeding up (the $n$-dependency in) these randomized algorithms needs faster deterministic algorithms. (B) The more surprising  direction is the reverse. Chang et al. \cite{chang2016exponential} recently showed that for a large class of problems the randomized complexity on $n$-node graphs is at least the deterministic complexity on $\Theta\paren{\sqrt{\log n}}$-node graphs. Hence, if one improves over (the $n$-dependency in) the current randomized algorithms, one has inevitably improved the corresponding deterministic algorithm. 

Ghaffari, Kuhn, and Maus \cite{ghaffari2016complexity} recently proved a \emph{completeness}-type result which shows that \emph{``the only obstacle"} for efficient deterministic distributed graph algorithms is deterministically \emph{rounding} fractional values to integral values while approximately preserving some linear constraints.\footnote{Stating this result in full generality requires some definitions. See \cite{ghaffari2016complexity} for the precise statement.} To put it more positively, if we find an efficient deterministic method for rounding, we would get efficient algorithms for essentially all the classic local graph problems, including the four mentioned above. Our results become more instructive when viewed in this context. The common denominator of our results is a deterministic distributed method which allows us to round fractional matchings to integral matchings. This can be more generally 
seen as rounding the fractional solutions of a special class of \emph{linear programs} (LPs) to integral solutions.  To the best of our knowledge, this is the first known \emph{deterministic distributed rounding} method. We can now say that

\begin{center}
\begin{minipage}{0.8\linewidth}
\vspace{-8pt}
\begin{mdframed}[hidealllines=true, backgroundcolor=gray!00]
\emph{matching admits an efficient deterministic algorithm because \\matching admits an efficient deterministic distributed rounding.}  
\vspace{-2pt}
\end{mdframed}
\end{minipage}
\end{center}

\subsection{Our Results}\label{subsec:results}
We provide improved distributed algorithms for a number of matching problems, as we overview next.
\subsubsection{Approximate Maximum Matching}
\begin{theorem}\label{2+eps-approx-MaxM-bipartite}
There is an $O\paren{\log^2 \Delta \cdot \log \frac{1}{\eps} + \log^* n}$-round deterministic distributed algorithm for a $(2+\eps)$-approximate maximum matching, for any $\eps>0$. 
\end{theorem}
\noindent There are three remarks in order, regarding this result:
\begin{itemize}
\item For constant $\eps>0$, this $O(\log^2 \Delta + \log^* n)$-round algorithm is significantly faster than the previously best known deterministic constant approximations, especially in low-degree graphs: the $O(\Delta + \log^* n)$-round $2$-approximation of Panconesi and Rizzi \cite{panconesirizzi2000}, the $O(\log^4 n)$-round $2$-approximation of Ha\'n\'ckowiak et al. \cite{hanckowiak1999faster}, the $O(\log^4 n)$-round $(3/2)$-approximation of Czygrinow et al. \cite{czygrinow2004distributed,czygrinow2004fast}, and its extension \cite{czygrinow2003distributed} which finds a $(1+\eps)$-approximation in $\log^{O\paren{\frac{1}{\eps}}} n$ rounds. 

\item This $O(\log^2 \Delta + \log^* n)$-round complexity gets close to the lower bound --- due to the celebrated results of Kuhn et al. \cite{DBLP:conf/soda/KuhnMW06,DBLP:journals/jacm/KuhnMW16} and Linial \cite{linial1987LOCAL} --- of $\Omega(\log \Delta/\log\log \Delta + \log^* n)$ that holds for any constant approximation of matching, even for randomized algorithms.

\item This distributed \local algorithm can be transformed to a deterministic \emph{Local Computation Algorithm} (LCA) \cite{Alon2012LCA, Rubinfeld2011LCA} for a $(2+\eps)$-approximation of maximum matching, with a query complexity of $2^{O(\log^3{\Delta})}\cdot \log^* n$. This is essentially by using the standard method of Parnas and Ron~\cite{parnas2007approximating}, with an additional idea of \cite{even2014deterministic}. Using slightly more care, the query complexity can be improved to $2^{O(\log^2{\Delta})} \cdot \log^* n$. Since formally stating this result requires explaining the computational model of LCAs, we defer that to the journal version. We remark that this query complexity improves almost exponentially over the previous deterministic constant approximations with $2^{\Omega(\Delta\cdot \log \Delta)} \cdot\log^* n$ \cite{even2014deterministic}.

\end{itemize}

\subsubsection{(Almost) Maximal Matching, and Edge Dominating Set}
\paragraph{Maximal Matching} 
Employing our approximation algorithm for maximum matching, we get an $O(\log^2 \Delta \cdot \log n)$-round deterministic distributed algorithm for maximal matching. 
\begin{theorem}\label{MM}
There is an $O(\log^2 \Delta \cdot \log n)$-round deterministic maximal matching algorithm.
\end{theorem}
This is the first improvement in about 20 years over the breakthroughs of Ha\'n\'ckowiak et al., which presented first an $O(\log^7 n)$- \cite{HanckowiakKP98} and then an $O(\log^4 n)$-round \cite{hanckowiak1999faster} algorithm for maximal matching.

As alluded to before, this improvement in the deterministic complexity directly leads to an improvement in the $n$-dependency of the randomized algorithms. In particular, plugging in our improved deterministic algorithm in the maximal matching algorithm of Barenboim et al. \cite{barenboim2012locality} improves their round complexity from $O(\log^4 \log n + \log\Delta)$ to $O(\log^3 \log n + \log \Delta)$. 

\begin{corollary}\label{Crl:RandMM}
There is an $O(\log^3 \log n + \log \Delta)$-round randomized distributed algorithm that with high probability\footnote{As standard, \emph{with high probability} means with probability at least $1-1/n^{c}$, for a desirably large constant $c\geq 2$.} computes a maximal matching.
\end{corollary}

\paragraph{Almost Maximal Matching} Recently, there has been quite some interest in characterizing the $\Delta$-dependency in the complexity of maximal matching, either with no dependency on $n$ at all or with at most an $O(\log^* n)$ additive term \cite{Hirvonen:2012,Goos:2014}. 
G\"{o}\"{o}s et al. \cite{Goos:2014} conjectured that 
\begin{center}
\emph{there should be no $o(\Delta) +O(\log^* n)$ algorithm for computing a maximal matching.} 
\end{center}
\Cref{MM} does not provide any news in this regard, because of its multiplicative $\log n$-factor. Indeed, our findings also seem to be consistent with this conjecture and do not suggest any way for breaking it. 
However, using some extra work, we can get a faster algorithm for $\eps$- maximal matching, a matching that leaves only $\eps$-fraction of edges among unmatched nodes, for a desirably small $\eps>0$.  

\begin{theorem}\label{eps-almost-MM-bipartite}
There is an $O\paren{\log^2 \Delta \cdot \log \frac{1}{\eps} + \log^* n}$-round deterministic distributed algorithm for an $\eps$-maximal matching, for any $\eps>0$. 
\end{theorem}
This theorem statement is interesting because of two aspects: (1) This faster almost maximal matching algorithm sheds some light on the difficulties of proving the aforementioned conjecture. In a sense, any conceivable proof of this conjectured lower bound must distinguish between maximal and almost maximal matchings and rely on the fact that precisely a maximal matching is desired, and not just something close to it. Notice that since the complexity of \Cref{eps-almost-MM-bipartite} grows slowly as a function of $\eps$, we can choose $\eps$ quite small. By setting $\eps = \Delta^{-\poly\log \Delta}$, we get an algorithm that, in $O(\poly\log \Delta + \log^* n)$ rounds, produces a matching that seems to be maximal for almost all nodes, even if they look up to their $\poly\log \Delta$-hop neighborhood. (2) Perhaps, in some practical settings, this almost maximal matching, which practically looks maximal for essentially all nodes, may be as useful as maximal matching, especially since it can be computed much faster. 

\paragraph{Approximate Minimum Edge Dominating Set} As a corollary of the almost maximal matching algorithm of \Cref{eps-almost-MM-bipartite}, we get a fast algorithm for approximating \emph{minimum edge dominating set}, which is the smallest set of edges such that any edge shares at least one endpoint with them. 
The proof appears in \Cref{sec:EDS}.  

\begin{corollary}\label{2+eps-approx-EDS}
There is an $O(\log^2\Delta \cdot \log \frac{\Delta}{\eps}+ \log^*n)$-round deterministic distributed algorithm for a $(2+\eps)$-approximate minimum edge dominating set, for any $\eps> 0$. 
\end{corollary}
Previously, the fastest algorithms ran in $O(\Delta + \log^* n)$ rounds \cite{panconesirizzi2000} or $O(\log^4 n)$ rounds \cite{hanckowiak1999faster}, providing $2$-approximations. Moreover, Suomela \cite{suomela2010EdgeDominatingSets} provided roughly $4$-approximations in $O(\Delta^2)$ rounds, in a restricted variant of the \local model with only port numberings.

\subsubsection{Approximate Maximum Weighted Matching and B-Matching}
An interesting aspect of the method we use is its flexibility and generality. In particular, the algorithm of \Cref{2+eps-approx-MaxM-bipartite} can be easily extended to computing a $(2+\eps)$-approximation of maximum weighted matching, and more interestingly, to a $(2+\eps)$-approximation of maximum weighted \emph{b-matching}. These extensions can be found in \Cref{sec:b-matching,sec:weightedmatchings}. 
\begin{theorem}\label{2+eps-approx-weighted-MaxM}
There is an $O(\log^2\Delta \cdot \log\frac{1}{\eps} + \log^*n)$-round deterministic distributed algorithm for a $(2+\eps)$-approximate maximum weighted matching, or $b$-matching, for any $\eps> 0$. 
\end{theorem}

To the best of our knowledge, this is the first distributed deterministic algorithm for approximating maximum (weighted) $b$-matching. Moreover, even in the case of standard matching, it improves over the previously best-known algorithm: A deterministic algorithm for $(6+\eps)$-approximation of maximum weighted matching was provided by Panconesi and Sozio \cite{panconesi2010fast}, with a round complexity of $O\paren{\log^4 n \cdot \log_{1+\eps} W}$, where $W$ denotes the maximum normalized weight. However, that deterministic algorithm does not extend to $b$-matching.  
\subsection{Related Work, Randomized Distributed Matching Approximation}
\vspace{-4pt}
Aside from the deterministic algorithms discussed above, there is a long line of research on randomized distributed approximation algorithms for matching: for the unweighted case, \cite{israeli1986fast} provide a $2$-approximation in $O(\log n)$ rounds, and \cite{lotker2008improved} a $(1+\eps)$-approximation in $O(\log n)$ for any constant $\eps>0$. For the weighted case, \cite{wattenhofer2004distributed, Lotker:2007, lotker2008improved} provide successively improved algorithms, culminating in the $O(\log \frac{1}{\eps} \cdot \log n)$-round $(2+\eps)$-approximation of \cite{lotker2008improved}. Moreover, \cite{koufogiannakis2009distributed} present an $O(\log n)$-round randomized algorithm for $2$-approximate weighted $b$-matching.
 
\vspace{-3pt}

\section{Our Deterministic Rounding Method, in a Nutshell}

The main ingredient in our results is a simple deterministic method for \emph{rounding} fractional solutions to integral solutions. We believe that this \emph{deterministic distributed rounding} will be of interest well beyond this paper. 
To present the flavor of our deterministic rounding method, here we overview it in a simple special case: we describe an $O(\log^2 \Delta)$-round algorithm for a constant approximation of the maximum unweighted matching in 2-colored bipartite graphs. The precise algorithm and proof appear in \Cref{specialCaseBipartite}. 

\paragraph{Fractional Solution} First, notice that finding a fractional approximate maximum matching is straightforward. In $O(\log \Delta)$ rounds, we can compute a fractional matching $\mathbf{x}\in [0,1]^{m}$ whose total value $\sum_{e} x_{e}$ is a constant approximation of maximum matching. One standard method is as follows: start with all edge values at $x_{e} = 2^{-\lceil\log{\Delta}\rceil}$. Then, for $O(\log \Delta)$ rounds, in each round raise all edge values $x_e$ by a $2$-factor, except for those edges that are incident to a node $v$ such that $\sum_{e\in E(v)} x_e \geq 1/2$. Throughout, $E(v):=\{e \in E \colon v \in e\}$ denotes the set of edges incident to node $v$. One can easily see that this fractional matching has total value $\sum_{e} x_{e}$ within a $4$-factor of the maximum matching.

\paragraph{Gradual Rounding} We gradually round this fractional matching $\mathbf{x}\in [0,1]^{m}$ to an integral matching $\mathbf{x'}\in \{0,1\}^{m}$ while ensuring that we do not lose much of the value, i.e., $\sum_{e} x'_{e} \geq (\sum_{e} x_{e})/C$, for some constant $C$. We have $O(\log \Delta)$ rounding phases, each of which takes $O(\log \Delta)$ rounds. In each phase, we get rid of the smallest (non-zero) values and thereby move closer to integrality. The initial fractional matching has\footnote{Any fractional maximum matching can be transformed to this format, with at most a $2$-factor loss in the total value: simply round down each value to the next power of $2$, and then drop edges with values below $2^{-(\lceil\log{\Delta}\rceil+1)}$. } only values $x_{e} = 2^{-i}$ for $i\in \{0, \dotsc, \lceil\log{\Delta}\rceil\}$  or $x_e=0$. In the $k^{th}$ phase, we partially round the edge values $x_e = 2^{-i}$ for $i=\lceil\log{\Delta}\rceil - k +1$. Some of these edges will be raised to $x_{e} = 2 \cdot 2^{-i}$, while others are dropped to $x_{e}=0$. The choices are made in a way that keeps $\sum_{e} x_{e}$ essentially unchanged, as we explain next.

Consider the graph $H$ edge-induced by edges $e$ with value $x_e = 2^{-i}$. For the sake of simplicity, suppose all nodes of $H$ have even degree. Dealing with odd degrees requires some delicate care, but it will not incur a loss worse than an $O\paren{2^{-i}}$-fraction of the total value. In this even-degree graph $H$, we effectively want that for each node $v$ of $H$, half of its edges raise $x_e = 2^{-i}$ to $x_e=2\cdot 2^{-i}$ while the others drop it to $x_e =0$. For that, we generate a degree-$2$ graph $H'$ by replacing each node $v$ of $H$ with $d_{H}(v)/2$ nodes, each of which gets two of $v$'s edges\footnote{This simple idea has been used frequently before. For instance, it gives an almost trivial proof of Petersen's 2-factorization theorem from 1891 \cite{mulder1992julius}. It has also been used by \cite{israeli1986improved, HanckowiakKP98,hanckowiak1999faster}.}. Notice that the edge sets of $H'$ and $H$ are the same. Graph $H'$ is simply a set of cycles of even length, as $H$ was bipartite.

In each cycle of $H'$, we would want that the raise and drop of edge weights is alternating. That is, odd-numbered, say, edges are raised to $x_e=2\cdot 2^{-i}$ while even-numbered edges are dropped to $x_e=0$. This would keep $\mathbf{x}$ a valid fractional matching--- meaning that each node $v$ still has $\sum_{e\in E(v)} x_e \leq 1$--- because the summation $\sum_{e\in E(v)} x_e$ does not increase, for each node $v$. 
Furthermore, it would keep the total weight $\sum_{e} x_e$ unchanged. If the cycle is shorter than length $O(\log \Delta)$, this raise/drop sequence can be identified in $O(\log \Delta)$ rounds. For longer cycles, we cannot compute such a perfect alternation in $O(\log \Delta)$ rounds. However, one can do something that does not lose much\footnote{Our algorithm actually does something slightly different, but describing this ideal procedure is easier.}: imagine that we chop the longer cycles into edge-disjoint paths of length $\Theta(\log \Delta)$. In each path, we drop the endpoints to $x_e=0$ while using a perfect alternation inside the path. These border settings mean we lose $\Theta(1/\log \Delta)$-fraction of the weight. Thus, even over all the $O(\log \Delta)$ iterations, the total loss is only a small constant fraction of the total weight.

\section{Preliminaries}\label{preliminaries}
\vspace{-5pt}
\paragraph{Matching and Fractional Matching}
An integral matching $M$ is a subset of $E$ such that $e \cap e' = \emptyset$ for all $e\neq  e' \in M$. It can be seen as an assignment of values $x_e \in \{0,1\}$ to edges, where $x_e=1$ iff $e \in M$, such that $c_v:=\sum_{e \in E(v)} x_e \leq 1$ for all $v \in V$. When the condition $x_e \in \{0,1\}$ is relaxed to $0 \leq x_e \leq 1$, such an assignment is called a fractional matching. 

\paragraph{B-Matching} A $b$-matching for $b$-values $\{1 \leq b_v \leq d_G(v) \colon v \in V\}$ is an assignment of values $x_e \in \{0,1\}$ to edges $e \in E$ such that $\sum_{e \in E(v)} x_e \leq b_v$ for all $v \in V$. 
Again, one can relax this to fractional $b$-matchings by replacing $x_e \in \{0,1\}$ with $0 \leq x_e \leq 1$. 

\paragraph{Maximal and $\eps$-Maximal Matching} An integral matching is called maximal if we cannot add any edge to it without violating the constraints. For $\eps > 0$, we say that $M\subseteq E$ is an $\eps$- maximal matching if $|\Gamma^+(M)|\geq (1-\eps)|E|$ for $\Gamma^+(M):=\{e \in E \mid \exists e' \in M \colon e \cap e' \neq \emptyset\}$, that is, if after removing the edges in and incident to $M$ from $G$, at most $\eps |E|$ edges remain. 

\paragraph{Maximum and Approximate Maximum Matching} A matching $M^*$ is called maximum if it is the/a largest matching in terms of cardinality. 
For any $c > 1$, we say that a matching is $c$-approximate if $c \sum_{e\in E} x_e\geq |M^*|$ for a maximum matching $M^*$. 
In a weighted graph where each edge $e$ is assigned a weight $w_e \geq 0$, we say that $M^*$ is a maximum weighted matching if it is the/a matching with maximum weight $w(M^*):=\sum_{e \in M^*} w_e$. An integral matching $M$ is a $c$-approximate weighted matching if $  c \sum_{e \in M} w_e\geq w(M^*)$. 

\vspace{0.2cm}
We now state some simple and well-known facts about matchings.

\begin{lemma}\label{trivial-size-MM}
For a maximal matching $M$ and a maximum matching $M^*$ in $G=(V,E)$, we have the following two properties: 
(i) $|M|\geq \frac{|E|}{2 \Delta -1}$, and (ii) $\frac{|M^*|}{2} \leq |M| \leq |M^*|$.
\end{lemma}

\begin{lemma}[Panconesi and Rizzi \cite{panconesirizzi2000}]\label{MM-bipartite-small-degree} There is an $O(\Delta + \log^* n)$-round deterministic distributed algorithm for maximal matching. Furthermore, if a $q$-coloring of the graph is provided, then the algorithm runs in $O(\Delta + \log^* q)$ rounds. 
\end{lemma}

Many problems are easier in small-degree graphs. To exploit this fact, we sometimes use the following simple transformation which decomposes a graph into graphs with maximum degree 2 --- that is, node-disjoint paths and cycles --- with the same edge set, in zero rounds. As mentioned before, this has been used frequently in prior work \cite{mulder1992julius, israeli1986improved, HanckowiakKP98,hanckowiak1999faster}.

\begin{figure}[t]
\centering
\includegraphics[width=0.4\textwidth]{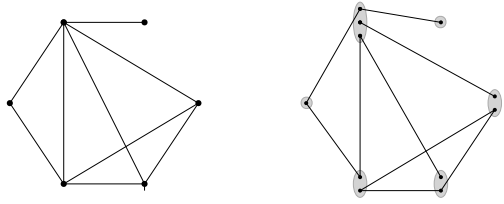}
\caption{A graph and its 2-decomposition.}
\label{fig:2decomp}
\end{figure}

\paragraph{2-decomposition} We \emph{$2$-decompose} graph $G$ as follows. For every node $v \in V$, introduce $\ceil{\frac{d_G(v)}{2}}$ copies and arbitrarily split its incident edges among these copies in such a way that every copy has degree 2, with the possible exception of one copy which has degree $1$ (when $v$ has odd degree). The graph on these copy nodes is what we call a $2$-decomposition of $G$. See \Cref{fig:2decomp} for an example.

\section{Approximate Maximum Matching}\label{sec:approx-MaxM}
We present a $(2+\eps)$-approximation algorithm for maximum matching, proving \Cref{2+eps-approx-MaxM-bipartite}. The first step towards this goal is finding a constant approximation, explained in \Cref{sec:bipartite}. We show in \Cref{wrapup} how to further improve this approximation ratio to $2+\eps$. 

\subsection{Constant Approximate Maximum Matching}\label{sec:bipartite}
In this subsection, we show how to compute a constant approximation. 

\begin{lemma}\label{const-approx-MaxM}
There is an $O(\log^2\Delta+\log^*n)$-round deterministic distributed algorithm for a $c$-approximate maximum matching, for some constant $c$.
\end{lemma}

The key ingredient for our $c$-approximation algorithm of \Cref{const-approx-MaxM} is a distributed algorithm that computes a constant approximate maximum matching in the special case of a $2$-colored bipartite graph. We first present the algorithm for this special case in \Cref{specialCaseBipartite}, and then explain in \Cref{sec:general} how to reduce the general graph case to the bipartite case, hence proving \Cref{const-approx-MaxM}. 

\subsubsection{Constant Approximate Maximum Matching in Bipartite Graphs}\label{specialCaseBipartite}
Next, we show how to find a $c$-approximate matching in a 2-colored bipartite graph. 
\begin{lemma}\label{const-approx-MaxM-bipartite}
There is an $O(\log^2\Delta)$-round deterministic distributed algorithm for a $c$-approximate maximum matching in a 2-colored bipartite graph, for some constant $c$.
\end{lemma} 

\paragraph{Roadmap}  The proof of \Cref{const-approx-MaxM-bipartite} is split into three parts. In the first step, explained in \Cref{fractional}, we compute a $2^{-\ceil{\log\Delta}}$-fractional 4-approximate maximum matching in $O(\log\Delta)$ rounds. The second step, which is also the main step of our method and is formalized in \Cref{rounding-phase}, is a method to round these fractional values to almost integrality in $O(\log^2\Delta)$ rounds. In the third step, presented in \Cref{rounding-end}, we resort to a simple constant-round algorithm to transform the almost integral matching that we have found up to this step into an integral matching. As a side remark, we note that we explicitly state some of the constants in this part of the paper, for the sake of readability. We remark that these constants are not the focus of this work, and we have not tried to optimize them. 

We start with some helpful definitions.

\begin{definition}[{Loose and tight nodes and edges}] Given a fractional matching, we call a node $v$ \emph{loose} if $c_v =\sum_{e \in E(v)} x_e \leq  \frac{1}{2}$, and \emph{tight} otherwise, where $E(v):=\{e \in E \colon v \in e\}$. We call an edge \emph{loose} if both of its endpoints are loose; otherwise, the edge is called \emph{tight}.
\end{definition}

\begin{definition}[{The fractionality of a fractional matching}] We call a fractional matching $2^{-i}$-fractional for an $i \in \mathbb{N}$ if $x_e \in \{0\} \bigcup \set{ 2^{-j} \colon 0 \leq j \leq i}$. Notice that a $2^{-0}$-fractional matching is simply an integral matching. 
\end{definition}

\paragraph{Step 1, Fractional Matching} We show that a simple greedy algorithm already leads to a fractional 4-approximate maximum matching.

\begin{lemma}\label{fractional}
There is an $O(\log\Delta)$-round deterministic distributed algorithm for a $2^{-\ceil{\log \Delta}}$-fractional $4$-approximate maximum matching. 
\end{lemma}
\begin{proof}
Initially, set $x_e =2^{-\ceil{\log \Delta}}$ for all $e \in E$. This trivially satisfies the constraints $c_v =\sum_{e \in E(v)} x_e \leq 1$. Then, we iteratively raise the value of all loose edges in parallel by a $2$-factor. This can be done in $O(\log \Delta)$ rounds, since at the latest when the value of an edge is $1/2$, both endpoints would be tight. Once all edges are tight, for a maximum matching $M^*$ we have $\sum_{e\in E} x_e=\frac{1}{2}\sum_{v \in V} c_v \geq \frac{1}{2}\sum_{e=\{u,v\} \in M^*} (c_u+ c_v) >\frac{|M^*|}{4}$ .
\end{proof}

\paragraph{Step 2, Main Rounding} The heart of our approach, the Rounding Lemma, is a method that successively turns a $2^{-i}$-fractional matching into a $2^{-i+1}$-fractional one, for decreasing values of $i$, while only sacrificing the approximation ratio by a little. 

\begin{lemma}[\textbf{Rounding Lemma}]\label{rounding-phase}
There is an $O\paren{\log^2 \Delta}$-round deterministic distributed algorithm that transforms a $2^{-\ceil{\log\Delta}}$-fractional $4$-approximate maximum matching in a $2$-colored bipartite graph into a $2^{-4}$-fractional $14$-approximate maximum matching.\end{lemma}
\begin{proof}
Iteratively, for $k=1, \dotsc, \ceil{\log \Delta}-4$, in phase $k$, we get rid of edges $e$ with value $x_e=2^{-i}$ for $i=\ceil{\log \Delta}-k+1$ by either increasing their values by a 2-factor to $x_{e}=2^{-i+1}$ or setting them to $x_e=0$. In the following, we describe the process for one phase $k$, thus a fixed $i$. 

Let $H$ be the graph induced by the set $E_i:=\{ e \in E \colon x_e = 2^{-i}\}$ of edges with value $2^{-i}$ and use $H'$ to denote its 2-decomposition. Notice that $H'$ is a node-disjoint union of paths and even-length cycles. Set $\ell=12\log \Delta$. We call a path/cycle \emph{short} if it has length at most $\ell$, and \emph{long} otherwise. We now process short and long cycles and paths, by distinguishing three cases, as we discuss next. Each of these cases will be done in $O(\log \Delta)$ rounds, which implies that the complexity of one phase is $O(\log \Delta)$. Thus, over all the $O(\log \Delta)$ phases, this rounding algorithm takes $O(\log^2 \Delta)$ rounds.

\paragraph{Case A, Short Cycles}
Alternately set the values of the edges to 0 and to $2^{-i+1}$. Since the cycle has even length, the values $c_v=\sum_{e \in E(v)} x_e$ for all nodes $v$ in the cycle remain unaffected by this update. Moreover, the total value of the edges in the cycle stays the same. 

\paragraph{Case B, Long Cycles and Long Paths}
We first orient the edges in a manner that ensures that each maximal directed path has length at least $\ell$. This is done in $O(\ell)$ rounds. For that purpose, we start with an arbitrary orientation of the edges. Then, for each $j=1, \dotsc, \ceil{\log \ell}$, we iteratively merge two (maximal) directed paths of length $<2^j$ that are directed towards each other by reversing the shorter one, breaking ties arbitrarily. 
For more details of this orientation step, we refer to \cite[Fact 5.2]{HanckowiakKaronskiPanconesi2001MM}.

Given this orientation, we determine the new values of $x_e$ as follows. Recall that we are given a $2$-coloring of nodes. Set the value of all border edges (that is, edges that have an incident edge such that they are either oriented towards each other or away from each other) to 0, increase the value of a non-border edge to $2^{-i+1}$ if it is oriented towards a node of color 1, say, and set it to 0 otherwise.

\begin{figure}[t]
\centering
\includegraphics[width=.9\textwidth]{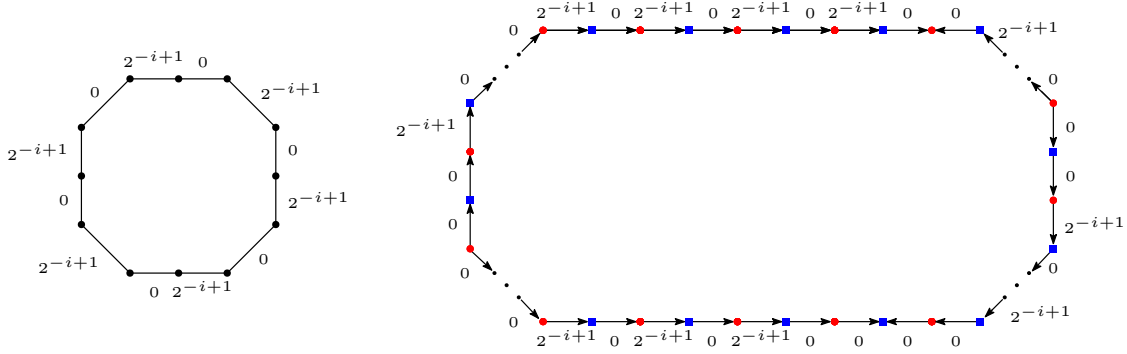}
\caption{The edge values of a short and a long cycle induced by edges in $E_i$ after rounding. In the long cycle, nodes of color 1 are depicted as blue squares and nodes of color 2 as red circles.}
\label{fig:roundingLong}
\end{figure}

\begin{figure}[h]
\centering
\includegraphics[width=0.4\textwidth]{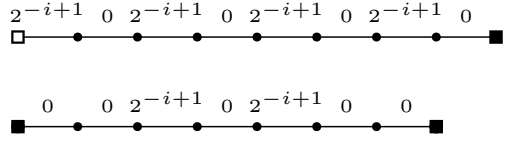}
\caption{The edge values of two short paths induced by edges in $E_i$ after rounding. Tight endpoints are depicted as (unfilled) boxes and loose endpoints as (filled) squares.}
\label{fig:roundingShort}
\end{figure}

Now, we show that this process generates a valid fractional matching while incurring only a small loss in the value. Observe that no constraint is violated, as for each node the value of at most one incident edge can be raised to $2^{-i+1}$ while the other is dropped to 0. 
Moreover, in each maximal directed path, we can lose at most $3\cdot 2^{-i}$ in the total sum of edge values. This happens in the case of an odd-length path starting with a node of color 2. Hence, we lose at most a $\frac{3}{\ell}$-fraction of the total sum of the edge values in long cycles and long paths.

\vspace{-10pt}
\paragraph{Case C, Short Paths}
Give the path an arbitrary direction, that is, identify the first and the last node. Set the value of the first edge to $2^{-i+1}$ if the first node is loose, and to $0$ otherwise. Alternately, starting with value $0$ for the second edge, set the value of every even edge to 0 and of every odd edge to $2^{-i+1}$. If the last edge should be set to $2^{-i+1}$ (that is, the path has odd length) but the last node is tight, set the value of that last edge to 0 instead. 

If a node $v$ is in the interior of the path, that is, not one of the endpoints, then $v$ can have at most one of its incident edges increased to $2^{-i+1}$ while the other one decreases to 0. Hence the summation $c_v= \sum_{e\in E(v)}x_e$ does not increase. If $v$ is the first or last node in the path, the value of the edge incident to $v$ is increased only if $v$ was loose, i.e., if $c_v= \sum_{e\in E(v)}x_e \leq \frac{1}{2}$. In this case, we still have $c_v\leq 1$ after the increase, as the value of the edge raises by at most a $2$-factor.

We now argue that the value of the matching has not decreased by too much during this update. For that, we group the edges into blocks of two consecutive edges, starting from the first edge. If the path has odd length, the last block consists of a single edge. The block value, that is, the sum of the values of its two edges, of every interior (neither first nor last) block is unaffected. If an endpoint $v$ of a path is loose, the value of the block containing $v$ remains unchanged or increases (in the case of an odd-length path ending in $v$). If $v$ is tight, then the value of its block stays the same or decreases by $2^{-i+1}$, which is at most a $2^{-i+2}$-fraction of the value $c_v$. This allows us to bound the loss in terms of these tight endpoints. The crucial observation is that every node can be endpoint of a short path at most once. This is because, in the 2-decomposition, a node can be the endpoint of a path only if it has a degree-1 copy, which happens only for odd-degree vertices and then exactly once. Thus, we lose at most a $2^{-i+2}$-fraction in $\sum_{v \in V} c_v$ when updating the values in short paths.

\paragraph{Analyzing the Overall Effect of Rounding}
First, we show that over all the rounding phases, the overall loss is only a constant fraction of the total value $\sum_{e\in E} x_e$. 
Let $x_e^{(i)}$ and $c_v^{(i)}$ denote the value of edge $e$ and node $v$, respectively, before eliminating all the edges with value $2^{-i}$. Putting together the loss analyses discussed above, we get
\begin{equation*}\begin{aligned}\sum_{e \in E} x_e^{(i-1)} &\geq  \sum_{e \in E} x_e^{(i)} - \frac{3}{\ell}\sum_{e\in E} x_e^{(i)}-2^{-i+2}\sum_{v \in V}c_v^{(i)} \geq \paren{ 1 - \frac{3}{\ell} - 2^{-i+3}} \sum_{e \in E} x_e^{(i)}.\end{aligned}\end{equation*} 
It follows that
\begin{equation*}\begin{aligned}\sum_{e \in E} x_e^{(4)} &\geq \paren{\prod_{i=5}^{\ceil{\log\Delta}} \paren{1-\frac{3}{\ell} - 2^{-i+3}}}\sum_{e \in E} x_e^{(\ceil{\log \Delta})} \geq \paren{\prod_{i=5}^{\ceil{\log \Delta}} e^{-2 \paren{\frac{3}{\ell} + 2^{-i+3}}}}\sum_{e \in E} x_e^{(\ceil{\log \Delta})} 
\\& \geq e^{-\frac{1}{4} -16\sum_{i=5}^{\ceil{\log \Delta}} 2^{-i}}\sum_{e \in E} x_e^{(\ceil{\log \Delta})} \geq \frac{1}{e^{\frac{5}{4}}}\sum_{e \in E} x_e^{(\ceil{\log \Delta})}\geq \frac{1}{4e^{\frac{5}{4}}} |M^*| \geq \frac{1}{14}|M^*|
\end{aligned}\end{equation*}
for a maximum matching $M^*$, recalling that we started with a 4-approximate maximum matching. Here, the second inequality holds because $\frac{3}{\ell} + 2^{-i+3} \leq \frac{1}{2}$, as $i \geq 5$. 
Finally, observe that in all the rounding phases the constraints $c_v=\sum_{e\in E(v)} x_{e} \leq 1$ are preserved, since the value $c_v$ can increase by at most a $2$-factor and only when $v$ is loose.
\end{proof}

\paragraph{Step 3, Final Rounding} So far, we have an almost integral matching. Next, we round all edges to either 0 or 1, by finding a maximal matching in the graph induced by edges with positive value.  

\begin{lemma}\label{rounding-end}
There is an $O(1)$-round deterministic distributed algorithm that, given a $2^{-4}$-fractional $14$-approximate maximum matching in a 2-colored bipartite graph, computes an integral matching that is $434$-approximate.
\end{lemma}
\begin{proof}
In the given $2^{-4}$-fractional matching, $x_e \neq 0$ means $x_e\geq \frac{1}{16}$. Thus, a node cannot have more than 16 incident edges with non-zero value in this fractional matching. In this constant-degree subgraph, a maximal matching $M$ can be found in $O(1)$ rounds using the algorithm in \Cref{MM-bipartite-small-degree}, recalling that we are given a 2-coloring. We have $|M| \geq \frac{|\{e \in E \colon x_e > 0\}|}{31} \geq\frac{1}{31} \sum_{e \in E} x_e$ by \Cref{trivial-size-MM} (i), and, since we started with a $14$-approximation, $M$ is $434$-approximate.   
\end{proof}

\subsubsection{Constant Approximate Maximum Matching in General Graphs}\label{sec:general}
We explain how the approximation algorithm for maximum matchings in 2-colored bipartite graphs can be employed to find approximate maximum matchings in general graphs. The main idea is to transform the given general graph into a bipartite graph with the same edge set in such a way that a matching in this bipartite graph can be easily turned into a matching in the general graph. 

\begin{proof}[Proof of \Cref{const-approx-MaxM}]
Let $\overrightarrow{E}$ be an arbitrary orientation of the edges $E$. Split every node $v\in V$ into two siblings $v_{\text{in}}$ and $v_{\text{out}}$, and add an edge $\{u_{\text{out}}, v_{\text{in}}\}$ to $E_B$ for every oriented edge $(u,v)\in \overrightarrow{E}$. Let $V_{\text{in}}:=\{ v_{\text{in}} \colon v \in V\}$ and $V_{\text{out}}:=\{ v_{\text{out}} \colon v \in V\}$ be the nodes with color 1 and 2, respectively. By \Cref{const-approx-MaxM-bipartite}, a $c$-approximate maximum matching $M_B$ in the bipartite graph $B=(V_{\text{in}} \bigcup V_{\text{out}},E_B)$ can be computed in $O\paren{\log^2 \Delta}$ rounds. We now go back to $V$, that is, merge $v_{\text{in}}$ and $v_{\text{out}}$ back into $v$. This makes the edges of $M_B$ incident to $v_{\text{in}}$ or $v_{\text{out}}$ now be incident to $v$, leaving us with a graph $G'=(V,M_B)\subseteq G$ with maximum degree $2$. 

We compute a maximal matching $M'$ in $G'$. Using the algorithm of \Cref{MM-bipartite-small-degree}, this can be done in $O(\log ^* n)$ rounds. If an $\poly \Delta$-coloring of $G$ is provided, which implies a coloring of $G'$ with $\poly \Delta$ colors, the round complexity of this step is merely $O(\log^*\Delta)$.

It follows from \Cref{trivial-size-MM} (i) that $|M'|\geq \frac{|M_B|}{3} \geq  \frac{|M_B^*|}{3c}\geq \frac{|M^*|}{3c}$ for maximum matchings $M_B^*$ in $B$ and $M^*$ in $G$, respectively. Thus, $M'$ is a $3c$-approximate maximum matching in $G$. The last inequality is true since by introducing additional nodes but leaving the edge set unchanged (when going from $G$ to $B$), the maximum matching size cannot decrease. 
\end{proof}

\subsection{Wrap-Up: $(2+\eps)$-Approximate Matching and Maximal Matching}\label{wrapup}
In this section, we iteratively invoke the constant approximation algorithm from the \Cref{sec:bipartite} to obtain algorithms for a $(2+\eps)$-approximate maximum matching (\Cref{2+eps-approx-MaxM-bipartite}) and a maximal matching (\Cref{MM}). 

The approximation ratio of a matching algorithm can be improved from $c$ to $2 + \eps$ easily, by $O\paren{\log \frac{1}{\eps}}$ repetitions: each time, we apply the algorithm of \Cref{const-approx-MaxM} to the remaining graph, and remove the found matching together with its neighboring edges from the graph. 

Before explaining the details, we present the following frequently used trick.
\begin{remark}\label{precompute} If a $\poly \Delta$-coloring of a graph is provided, we can go around the $\Omega(\log^* n)$ lower bound of Linial \cite{linial1987LOCAL}, omitting the additive $O(\log^* n)$ term from the round complexity of the algorithms presented in this paper. More generally, if such an algorithm is invoked iteratively, one can first precompute an $O(\Delta^2)$-coloring in $O(\log^*n )$ rounds using Linial's algorithm \cite{linial1992locality}, which allows us to replace the $O(\log^*n)$ term by $O(\log^* \Delta)$ by \Cref{MM-bipartite-small-degree} in each iteration. 
\end{remark}

\begin{proof}[Proof of \Cref{2+eps-approx-MaxM-bipartite}]
Starting with $G_0=G$, for $i=0, \dotsc, k-1$, where $k=O\paren{\log \frac{1}{\eps}}$, iteratively compute a $c$-approximate maximum matching $M_i$ in $G_i$, using the algorithm of \Cref{const-approx-MaxM}.
We delete $M_i$ together with its incident edges from the graph, that is, set $G_{i+1}=\paren{V, E(G_i) \setminus \Gamma^+(M_i)}$. 

Now, we argue that the obtained matching $\bigcup_{i=0}^{k-1}M_i$ is $(2+\eps)$-approximate. To this end, we bound the size of a maximum matching in the remainder graph $G_k$.

Let $M_i^*$ be a maximum matching in $G_i$. An inductive argument shows that $|M_i^*| \leq \paren{1 - \frac{1}{c}}^i |M^*|$. Indeed, observe 
$|M_{i+1}^*|\leq |M_i^*|-|M_i|\leq\paren{1- \frac{1}{c}}|M_i^*|\leq\paren{1-\frac{1}{c}}^{i+1}|M^*|$, where the first inequality holds since otherwise $M_{i+1}^* \cup M_i$ would be a better matching than $M_i^*$ in $G_i$, contradicting the latter's optimality. For $k=\log_{1-\frac{1}{c}} \frac{\eps}{2(2+\eps)}$, we thus have $|M_k^*|\leq \frac{\eps}{2(2+\eps)} |M^*|$. As $\bigcup_{i=0}^{k-1} M_i$ is a maximal matching in $G\setminus G_k$ by construction, $\paren{\bigcup_{i=0}^{k-1} M_i} \cup M_k^*$ is a maximal matching in $G$. By \Cref{trivial-size-MM} (ii), this means that $\card{\bigcup_{i=0}^{k-1} M_i} + |M_k^*| \geq \frac{|M^*|}{2}$, hence $\card{\bigcup_{i=0}^{k-1} M_i}\geq \left(\frac{1}{2}- \frac{\eps}{2(2+\eps)}\right)|M^*|\geq \frac{|M^*|}{2 + \eps}$. 

We have $O\paren{\log \frac{1}{\eps}}$ iterations, each taking $O(\log^2\Delta + \log^*n)$ rounds. As mentioned in \Cref{precompute}, by precomputing an $O(\Delta^2)$-coloring in $O(\log^*n)$ rounds, the round complexity of each iteration can be decreased to $O(\log^2 \Delta + \log^*\Delta)=O(\log^2 \Delta)$, leading to an overall running time of $O(\log^2\Delta \cdot \log \frac{1}{\eps} + \log^*n)$ rounds. 
\end{proof}

 \begin{remark}\label{remark-remaining-matching-size}
The analysis above shows that the matching $M$ computed by the algorithm of \Cref{2+eps-approx-MaxM-bipartite} is not only $(2+\eps)$-approximate, but also has the property that any matching in the remainder graph (induced by $E \setminus \Gamma^+(M)$) can have size at most $\eps|M^*|$ for a maximum matching $M^*$ in $G$.    
\end{remark}

If one increases the number of repetitions to $O(\log n)$, the found matching is maximal.   
\begin{proof}[Proof of \Cref{MM}]
Apply the $c$-approximation algorithm of \Cref{const-approx-MaxM} for $k=\log_{1- \frac{1}{c}} \frac{1}{n}$ iterations on the respective remainder graph, as described in the proof of \Cref{2+eps-approx-MaxM-bipartite}. The same analysis (also adopting the notation from there) shows that a maximum matching $M_k^*$ in the remainder graph $G_k$ must have size $|M_k^*| \leq \frac{|M^*|}{n}<1$, which means that $G_k$ is an empty graph. But then $\bigcup_{i=1}^{k-1} M_i$ must be maximal.  
\end{proof}

 \newpage
\section{Almost Maximal Matching}\label{sec:almostMM}
In the previous section, we have seen how one can obtain a matching that reduces the size of the matching in the remainder graph, that is, the graph after removing the matching and all incident edges, by a constant factor. Intuitively, one would expect that this also reduces the number of remaining edges by a constant factor, which would directly lead to an (almost) maximal matching just by repetitions. However, this is not the case, since not every matched edge removes the same number of edges from the graph, particularly in non-regular graphs. This calls for an approach that weights edges incident to nodes of different degrees differently, which naturally brings into play weighted matchings.

In \Cref{const-approximate-MaxWM}, we present a fast algorithm that finds a constant approximation of maximum weighted matching based on the algorithm of \Cref{2+eps-approx-MaxM-bipartite}. Then, we use this algorithm, by assigning certain weights to the edges, to find a matching that removes a constant fraction of the edges in \Cref{const-almost-MM}. Via  $O\paren{\log \frac{1}{\eps}}$ repetitions of this, each time removing the found matching and its incident edges, we get an $\eps$-maximal matching. More details are provided in the proof of \Cref{eps-almost-MM-bipartite} in the end of this section. Observe that when setting $\eps=\frac{1}{n^2}$, thus increasing the number of repetitions to $O(\log n)$, we obtain a maximal matching.
\begin{lemma}\label{const-approximate-MaxWM}
There is an $O\paren{\log^2 \Delta+\log^*n}$-round deterministic distributed algorithm for a $256$-approximate maximum weighted matching.  
\end{lemma}
\begin{proof}
We assume without loss of generality that the edge weights are normalized, that is, from a set $\{1, \dotsc, W\}$ for some maximum weight $W$. 
Round the weights $w_e$ for $e\in E$ down to the next power of $8$, resulting in weights $w_e'$. This rounding procedure lets us lose at most a $8$-factor in the total weight and provides us with a decomposition of $G$ into graphs $C_i=(V,E_i)$ with $E_i:=\{e \in E \colon w_e'=8^i\}$ for $i\in\{0, \dotsc, \floor{\log_8W}\}$. 

In parallel, run the algorithm of \Cref{2+eps-approx-MaxM-bipartite} with $\eps=1$ on every $C_i$ to find a $3$-approximate maximum matching $M_i$ in $C_i$ in $O(\log^2 \Delta+\log^*n)$ rounds. Observe that while the edges in $\bigcup_{i} M_i$ do not form a matching, since edges from $M_i$ and $M_j$ for $i\neq j$ can be neighboring, a matching $M \subseteq \bigcup_i M_i$ can be obtained by deleting all but the highest-index edge in every such conflict, that is, by removing all edges $e \in M_i$ with an incident edge $e' \in M_j$ for a $j > i$.

In the following, we argue that the weight of $M$ cannot be too small compared to the weight of $\bigcup_i M_i$ by an argument based on counting in two ways. 

Every edge $e \in \left(\bigcup_i M_i\right)\setminus M$ puts blame $w_e'$ on an edge in $M$ as follows. Since $e \notin M$, there is an edge $e'$ incident to $e$ such that $e \in M_i$ and $e' \in M_j$ for some $j>i$. If $e' \in M$, then $e$ blames weight $w_e$ on $e'$. If $e' \notin M$, then $e$ puts blame $w_e$ on the same edge as $e'$ does. 

For an edge $e \in M \cap E_i$ and $j \in [i]$, let $n_{j}$ be the maximum number of edges from $M_{i-j}$ that blame $e$. An inductive argument shows that $n_j\leq 2^j$. Indeed, there can be at most two edges from $M_{i-1}$ blaming $e$, at most one per endpoint of $e$, and, for $j > 1$, we have $n_j\leq 2 +  \sum_{j'=1}^{j-1} n_{j'}\leq 2 + \sum_{j'=1}^{j-1}2^{j'}= 2^j$, since at most two edges in $M_{i-j}$ can be incident to $e$ and at most one further edge can be incident to each edge in $M_{i-j'}$ for $j'<j$.

Therefore, overall, at most $\sum_{j=1}^i 2^j8^{i-j}\leq \frac{1}{3}8^{i} \leq \frac{1}{3}w_e'$ weight is blamed on $e$. This means that $\sum_{e\in \left(\cup_i M_i\right) \setminus M} w_e' \leq \frac{1}{3} \sum_{e \in M} w_e'$, hence  
$\sum_{e \in \cup_i M_i} w_e' \leq \frac{4}{3} \sum_{e \in M} w_e'$, and lets us conclude that $\sum_{e \in M^*} w_e\leq 8 \sum_{e \in M^*} w_e' \leq 24 \sum_{e \in \cup_i M_i} w_e' \leq 32 \sum_{e \in M} w_e' \leq 256 \sum_{ e \in M} w_e$ for a maximum weighted matching $M^*$.
\end{proof}

\newpage
Next, we explain how to use this algorithm to remove a constant fraction of edges, by introducing appropriately chosen weights. We define the weight of each edge to be the number of its incident edges. This way, an (approximate) maximum weighted matching corresponds to a matching that removes a large number of edges. 
\begin{lemma}\label{const-almost-MM}
There is an $O\paren{\log^2 \Delta+\log^*n}$-round deterministic distributed algorithm for a $\frac{511}{512}$-maximal matching.
\end{lemma}
\begin{proof}
For each edge $e=\{u,v\} \in E$, introduce a weight $w_e=d_G(u)+d_G(v)-1$, and apply the algorithm of \Cref{const-approximate-MaxWM} to find a $256$-approximate maximum weighted matching $M$ in $G$. 

For the weight $w(M^*)$ of a maximum weighted matching $M^*$, it holds that $w(M^*) \geq |E|$, as the following simple argument based on counting in two ways shows. Let every edge in $E$ put a blame on an edge in $M^*$ that is responsible for its removal from the graph as follows. An edge $e\in M^*$ blames itself. An edge $e \notin M^*$ blames an arbitrary incident edge $e' \in M^*$. Notice that at least one such edge must exist, as otherwise $M^*$ would not even be maximal. In this way, $|E|$ many blames have been put onto edges in $M^*$ such that no edge $e=\{u,v\}\in M^*$ is blamed more than $w_e$ times, as $e$ can be blamed by itself and any incident edge. Therefore, indeed $w(M^*)=\sum_{e\in M^*} w_e \geq |E|$, and, as $M$ is a $256$-approximate, it follows that $\sum_{e\in M} w_e \geq \frac{|E|}{256}$.  

Now, observe that $w_e$ is the number of edges that are deleted when removing $e$ together with its incident edges from $G$. Since every edge can be incident to at most two matched edges (and thus can be deleted by at most two edges in the matching), in total 
$|\Gamma^+(M)|\geq \frac{1}{2} \sum_{e\in M} w_e \geq \frac{|E|}{512}$ many edges are removed from $G$ when deleting the edges in and incident to $M$, which proves that $M$ is a 
$\frac{511}{512}$-maximal matching. 
\end{proof}

We iteratively invoke this algorithm to successively reduce the number of remaining edges.  

\begin{proof}[Proof of \Cref{eps-almost-MM-bipartite}]
For $i=0, \dotsc, k=O\paren{\log \frac{1}{\eps}}$ and $G_0=G$, iteratively apply the algorithm of \Cref{const-almost-MM} to $G_i$ to get a $c$-maximal matching $M_i$ in $G_i$. Set $G_{i+1}=\paren{V,E(G_i)\setminus \Gamma^+\paren{M_i}}$, that is, remove the matching and its neighboring edges from the graph. Then $M:= \bigcup_{i=0}^{k-1} M_i$ for $k=\log_{c} \eps$ is $\eps$-approximate, since $\card{E \setminus \Gamma^+(M)}=|E(G_k)|\leq c^k |E| \leq \eps |E|$, using $\card{E\paren{G_{i+1}}} \leq c \card{E(G_i)}$. 

Overall, recalling \Cref{precompute}, this takes $O( \log^2\Delta \cdot \log \frac{1}{\eps} + \log^*n)$.
\end{proof}

\section{Extensions and Corollaries}
\subsection{$B$-Matching}\label{sec:b-matching}
In this subsection, we explain that only slight changes to the algorithm of \Cref{sec:approx-MaxM} are sufficient to make it suitable also for computing approximations of maximum $b$-matching. To this end, we first introduce an approximation algorithm for maximum $b$-matching in 2-colored bipartite graphs in \Cref{const-approx-MaxBM-bipartite}. Then, we extend this algorithm to work for general graphs, in \Cref{const-approx-MaxBM}. Finally, in the second part of the proof of \Cref{2+eps-approx-MaxM-bipartite} presented at the end of this subsection, we show that the approximation ratio can be improved to a value arbitrarily close to 2, simply by repetitions of this constant approximation algorithm.  

\begin{lemma}\label{const-approx-MaxBM-bipartite}
There is an $O\paren{\log^2\Delta}$-round deterministic distributed algorithm for a $c$-approximate maximum $b$-matching in a 2-colored bipartite graph, for some constant $c$.
\end{lemma}

This result is a direct consequence of \Cref{fractionalB}, \Cref{rounding-phaseB}, and \Cref{rounding-end-B}, which we present next. These lemmas respectively show how a fractional constant approximate $b$-matching can be found, how this fractional matching can be round to almost integrality, and how these almost integral values can be turned into an integral matching, while only losing a constant fraction of the total value. The proofs are very similar to the ones in \Cref{sec:bipartite}, except for the very last step of rounding (\Cref{rounding-end-B}), which requires one extra step, as we shall discuss. 

In the following, we call a node $v\in V$ loose if $c_v=\sum_{e \in E(v)} x_e<  \frac{b_v}{2}$, and tight otherwise. As before, an edge $e$ is called \emph{tight} if either of its endpoints are tight, otherwise edge $e$ is called \emph{loose}. 

The next lemma shows how to obtain a $4$-approximate maximum $b$-matching in $O(\log \Delta)$ rounds. Alternatively, \cite{DBLP:conf/soda/KuhnMW06} find such a $b$-matching in $O(\log^2 \Delta)$ rounds. 
\begin{lemma}\label{fractionalB}
There is an $O(\log \Delta)$-round deterministic distributed algorithm for a $2^{-\ceil{\log \Delta}}$-fractional 4-approximate maximum $b$-matching. 
\end{lemma}
 \begin{proof}[Proof of \Cref{fractionalB}]
As in \Cref{fractional}, starting with $x_e =2^{-\ceil{\log \Delta}}$ (and thus $c_v \leq 1 \leq b_v$), in parallel, the edge values of non-tight edges with value $\leq \frac{1}{2}$ are gradually increased by a $2$-factor. This takes no more than $O(\log \Delta)$ rounds. We employ a simple argument based on counting in two ways to show that this yields a 4-approximation of a maximum $b$-matching $M^*$. Let each edge $e\in M^*$ blame one of its tight endpoints, if existent. If there is no tight endpoint, the value of the edge is $\leq 1$, and is blamed to $e$. In this way, each tight node $v$ --- which by definition has value $c_v = \sum_{e \in E(v)} x_e \geq \frac{1}{2}$ --- is blamed at most $b_v$ times. Let $v$ split this blame uniformly among its incident edges in $M^*$ such that each edge $e'$ is blamed at most twice its value $x_{e'}$. In this way, every edge $e'$ is blamed at most $4x_{e'}$, as it can be blamed by both of its tight endpoints, or by the edge itself if it has no tight endpoint. It follows that $|M^*| \leq 4 \sum_{e\in E} x_e$.
\end{proof}

Next, we transform this fractional solution into an almost integral solution, which is still a constant approximation. 

\begin{lemma}\label{rounding-phaseB}
There is an $O(\log^2 \Delta)$-round deterministic distributed algorithm that transforms a $2^{-\ceil{\log \Delta}}$-fractional $4$-approximate maximum $b$-matching in a $2$-colored bipartite graph into a $2^{-4}$-fractional $14$-approximate maximum $b$-matching.
\end{lemma}
\begin{proof}[Proof of \Cref{rounding-phaseB}]
As in the proof of \Cref{rounding-phase}, the edges of values $2^{-i}$ for $i=\ceil{\log \Delta}, \dotsc, 5$ are eliminated. We derive analogously that the fractional matching obtained at the end is a $14$-approximation, observing that changing the condition for tightness of a node from $c_v \geq \frac{1}{2}$ to $c_v \geq \frac{b_v}{2}\geq \frac{1}{2}$ only helps in the analysis. 
\end{proof}

In a final step, the almost integral solution is transformed into an integral one. Notice that for $b$-matchings, as opposed to standard matchings, the subgraph induced by edges with positive value need not have constant degree. In fact, a node $v \in V$ can have up to $16 b_v$ incident edges with non-zero value. This prevents us from directly applying the algorithm of \Cref{MM-bipartite-small-degree} to find a maximal matching in the subgraph with non-zero edge values, as this could take $O(\max_v b_v)=O(\Delta)$ rounds. 

\begin{lemma}\label{rounding-end-B}
There is an $O(1)$-round deterministic distributed algorithm that, given a $2^{-4}$-fractional $14$-approximate maximum $b$-matching in a $2$-colored bipartite graph, finds an integral $434$-approximate maximum $b$-matching.
\end{lemma}
\begin{proof}
We decompose the edge set induced by edges of positive value in the $2^{-4}$-fractional maximum $b$-matching $\left\{x_e^{(4)} \colon e \in E\right\}$ into constant-degree subgraphs $C_i=(V, E_i)$, as follows. We make at most $b_v$ copies of node $v$, and we arbitrarily split the edges among these copies in such a way that every copy has degree at most 16. This is done in a manner similar to the $2$-decomposition procedure.

In parallel, run the algorithm of \Cref{MM-bipartite-small-degree} on each $C_i$, in $O(1)$ rounds. This yields a maximal matching $M_i$ for each $C_i$ that trivially, by \Cref{trivial-size-MM} (i), satisfies the condition $|M_i| \geq \frac{|E_i|}{31}$. Now, let $M:=\bigcup_{i} M_i$. Since each node $v$ occurs in at most $b_v$ subgraphs and each $M_i$ is a matching in $C_i$, node $v$ cannot have more than $b_v$ incident edges in $M$. Thus, indeed, $M$ is a $b$-matching. Finally, observe that $M$ is $434$-approximate, since $|M|\geq \frac{|\{e \in E \colon x_e^{(4)} > 0\} |}{31}  \geq \frac{1}{31}\sum_{e \in E}x_e^{(4)}$. 
\end{proof}
 
A similar argument as in \Cref{const-approx-MaxM} shows that the algorithm for approximate maximum $b$-matchings in bipartite graphs from \Cref{const-approx-MaxBM-bipartite} can be adapted to work for general graphs.
\begin{lemma}\label{const-approx-MaxBM}
There is an $O\paren{\log^2\Delta+\log^*n}$-round deterministic distributed algorithm that computes a $c$-approximate maximum $b$-matching, for some constant $c$.   
\end{lemma}
\begin{proof}
Do the same reduction to a bipartite graph $B$ as in the proof of \Cref{const-approx-MaxM}, that is, create an in- and an out-copy of every node, and, for an arbitrary orientation of the edges, make each oriented edge incident to the respective copy of the corresponding nodes. 

Compute a $c$-approximate maximum $b$-matching $M_B$ in $B$ using the algorithm of \Cref{const-approx-MaxBM-bipartite}. Merging back the two copies of a node into one yields a graph with degree of node $v$ bounded by $2 b_v$, as $v_{\text{in}}$ and $v_{\text{out}}$ both can have at most $b_v$ incident edges in $M_B$. Now, compute a $2$-decomposition of this graph. On each component $C$ with edges $E_C\subseteq M_B$, find a maximal matching $M_C$ in $O(1)$ rounds by the algorithm of \Cref{MM-bipartite-small-degree}. 

Notice that for each node $v$ without a degree-1 copy, its degree is at least halved in $\bigcup_C M_C$ compared to $M_B$, and thus at most $b_v$. If a node $v$ has a degree-1 copy, then its degree need not be halved. But this can happen only if $v$'s degree in $M_B$ is odd, thus at most $2b_v-1$. In this case, $v$ has at most $b_v-1$ degree-2 copies and one degree-1 copy, which means that its degree in $\bigcup_C M_C$ is upper bounded by $b_v$. We conclude that $\bigcup_C M_C$ is indeed a $b$-matching.   

Moreover, it follows from $|M_C| \geq \frac{|E_C|}{3}$ by \Cref{trivial-size-MM} (i) that $\card{\bigcup_{C} M_C}\geq \frac{|M_B|}{3}\geq \frac{|M_B^*|}{3c} \geq \frac{|M^*|}{3c}$ for maximum $b$-matchings $M_B^*$ in $B$ and $M^*$ in $G$. Thus, $\bigcup_C M_C$ is $3c$-approximate.  
\end{proof}

\begin{proof}[Proof of \Cref{2+eps-approx-MaxM-bipartite} \textcolor{red}{for $b$-matching}]
Starting with $S_0=\emptyset$, $G_0=G$, and $b_v^0=b_v$ for all $v \in V$, for $i=0, \dotsc, k=O\paren{\log \frac{1}{\eps}}$, iteratively apply the algorithm of \Cref{const-approx-MaxBM} to $G_i$ with $b$-values $b_v^i$ 
to obtain a $c$-approximate maximum $b$-matching $M_i$ in $G_i$. Update $b_v^{i+1}=b_v^i-d_{M_i}(v)$ and $G_{i+1}=\paren{V, E_{i+1}}$ with $E_{i+1}:=E_i \setminus \paren{M_i \cup \left\{\{u,v\} \in E_i \colon b_v^{i+1}=0 \text{ or } b_u^{i+1} =0\right\}}$, that is, reduce the $b$-value of each vertex $v$ by the number $d_{M_i}(v)$ of incident edges in the matching $M_i$ and remove $M_i$ as well as all the edges incident to a node with remaining $b$-value 0 from the graph. The same analysis as in the proof of \Cref{2+eps-approx-MaxM-bipartite} for standard matchings in \Cref{wrapup} goes through and concludes the proof. 
\end{proof}

\begin{remark}\label{remark-remaining-bmatching-size}
The analysis above shows that the $b$-matching $M$ returned by the $b$-matching approximation algorithm of \Cref{2+eps-approx-MaxM-bipartite} is not only $(2+\eps)$-approximate in $G$, but also has the property that any $b$-matching in the remainder graph, after removing $M$ and all edges incident to a vertex $v$ with $b_v$ incident edges in $M$, can have size at most $\eps|M^*|$ for a maximum $b$-matching $M^*$ in $G$.    
\end{remark}

\subsection{Weighted Matching}\label{sec:weightedmatchings}
\begin{proof}[Proof Sketch of \Cref{2+eps-approx-weighted-MaxM}]
Using the idea from \cite[Section 4]{lotkerMatchingImproved}, we can iteratively invoke the constant-approximation algorithm of \Cref{const-approximate-MaxWM} $O(\log \frac{1}{\eps})$ times to get a $(2+\eps)$-approximate maximum weighted matching.

In each of the iterations $i \geq 1$, we set up a new auxiliary weighted graph as follows. Let $M_{i-1}$ be the matching obtained in the previous iteration. For every edge $e \in M_{i-1}$, let $w(e)=0$, and for every edge $e \notin M_{i-1}$, set $w(e)$ to the gain if $e$ is added to $M_{i-1}$ and the (possibly) incident edges in $M_{i-1}$ are deleted (if we lose by this change, we set $w(e)=0$). We then run the algorithm of \Cref{const-approximate-MaxWM} to get a $256$-approximate maximum weighted matching in this auxiliary graph, and augment $M_{i-1}$ along the edges in the matching in the auxiliary graph, i.e. add all edges in the matching to $M_{i-1}$ and remove all the possibly incident edges. Lemma 4.3 in \cite{lotkerMatchingImproved} (see also \cite{pettie2004simpler}) shows that then $w(M_i) \geq \frac{1}{2} \left(1-e^{-\frac{2}{3\cdot 256} i}\right)w(M^*)$. Thus, after $O(\log \frac{1}{\eps})$ iterations, a $(2+\eps)$-approximate maximum weighted matching is found.
\end{proof}

\subsection{Edge Dominating Set}\label{sec:EDS}
An edge dominating set is a set $D \subseteq E$ such that for every $e \in E$ there is an $e'\in D$ such that $e \cap e' \neq \emptyset$. A minimum edge dominating set is an edge dominating set of minimum cardinality. Since any maximal matching is an edge dominating set, an almost maximal matching can easily be turned into an edge dominating set: additionally to the edges in the almost maximal matching, add all the remaining (at most $\eps' |E|$ many) edges to the edge dominating set. When $\eps'$ is small enough, the obtained edge dominating set is a good approximation to the minimum edge dominating set. We next make this relation more precise. 

\begin{proof}[Proof of \Cref{2+eps-approx-EDS}] 
Apply the algorithm of \Cref{eps-almost-MM-bipartite} with $\eps'=\frac{\eps}{4 \Delta}$, say, to find an $\eps'$-maximal matching $M$ in $G$. It is easy to see that $D=M \cup \paren{E \setminus \Gamma^+(M)}$ is an edge dominating set. Moreover, due to the fact that a minimum maximal matching is a minimum edge dominating set (see e.g. \cite{YannakakisGavril1980EdgeDominatingSets}) and since maximal matchings can differ by at most a $2$-factor from each other, by \Cref{trivial-size-MM} (ii), it follows, also from \Cref{trivial-size-MM} (i), that $(2+\eps)|D^*|\geq \paren{1+ \frac{\eps}{2}}|M| \geq |M| + \frac{\eps}{2(2\Delta-1)}|E|> |M| + \eps'|E|\geq|D|$.
\end{proof}

\section*{Acknowledgment}
I want to thank Mohsen Ghaffari for suggesting this topic, for his guidance and his support, as well as for the many valuable and enlightening discussions. I am also thankful to Seth Pettie for several helpful comments.

\bibliographystyle{alpha}
%
\bibliography{ref2}

\begin{thebibliography}{ARVX12}

\bibitem[ABI86]{alon1986fast}
Noga Alon, L{\'a}szl{\'o} Babai, and Alon Itai.
\newblock A fast and simple randomized parallel algorithm for the maximal
  independent set problem.
\newblock {\em Journal of algorithms}, 7(4):567--583, 1986.

\bibitem[ARVX12]{Alon2012LCA}
Noga Alon, Ronitt Rubinfeld, Shai Vardi, and Ning Xie.
\newblock Space-efficient local computation algorithms.
\newblock In {\em Proceedings of the {ACM-SIAM} Symposium on Discrete
  Algorithms (SODA)}, pages 1132--1139, 2012.

\bibitem[BEPS12]{barenboim2012locality}
Leonid Barenboim, Michael Elkin, Seth Pettie, and Johannes Schneider.
\newblock The locality of distributed symmetry breaking.
\newblock In {\em Proceedings of the Symposium on Foundations of Computer
  Science (FOCS)}, pages 321--330, 2012.

\bibitem[CH03]{czygrinow2003distributed}
Andrzej Czygrinow and Micha{\l} Ha{\'n}{\'c}kowiak.
\newblock Distributed algorithm for better approximation of the maximum
  matching.
\newblock In {\em International Computing and Combinatorics Conference}, pages
  242--251, 2003.

\bibitem[CHS04a]{czygrinow2004distributed}
Andrzej Czygrinow, Micha{\l} Ha{\'n}{\'c}kowiak, and Edyta Szyma{\'n}ska.
\newblock Distributed algorithm for approximating the maximum matching.
\newblock {\em Discrete Applied Mathematics}, 143(1):62--71, 2004.

\bibitem[CHS04b]{czygrinow2004fast}
Andrzej Czygrinow, Micha{\l} Ha{\'n}{\'c}kowiak, and Edyta Szyma{\'n}ska.
\newblock A fast distributed algorithm for approximating the maximum matching.
\newblock In {\em Proceedings of the Annual European Symposium on Algorithms
  (ESA)}, volume 3221, pages 252--263, 2004.

\bibitem[CKP16]{chang2016exponential}
Yi{-}Jun Chang, Tsvi Kopelowitz, and Seth Pettie.
\newblock An exponential separation between randomized and deterministic
  complexity in the {LOCAL} model.
\newblock In {\em Proceedings of the Symposium on Foundations of Computer
  Science (FOCS)}, pages 615--624, 2016.

\bibitem[EMR14]{even2014deterministic}
Guy Even, Moti Medina, and Dana Ron.
\newblock Deterministic stateless centralized local algorithms for bounded
  degree graphs.
\newblock In {\em Proceedings of the Annual European Symposium on Algorithms
  (ESA)}, pages 394--405, 2014.

\bibitem[Gha16]{Ghaffari-MIS}
Mohsen Ghaffari.
\newblock An improved distributed algorithm for maximal independent set.
\newblock In {\em Proceedings of the {ACM-SIAM} Symposium on Discrete
  Algorithms (SODA)}, pages 270--277, 2016.

\bibitem[GHS14]{Goos:2014}
Mika G\"{o}\"{o}s, Juho Hirvonen, and Jukka Suomela.
\newblock Linear-in-delta lower bounds in the local model.
\newblock In {\em Proceedings of the {ACM} Symposium on Principles of
  Distributed Computing (PODC)}, pages 86--95, 2014.

\bibitem[GKM17]{ghaffari2016complexity}
Mohsen Ghaffari, Fabian Kuhn, and Yannic Maus.
\newblock On the complexity of local distributed graph problems.
\newblock In {\em Proceedings of the Symposium on Theory of Computing (STOC)},
  pages 784--797, 2017.

\bibitem[GS17]{GS17}
Mohsen Ghaffari and Hsin-Hao Su.
\newblock Distributed degree splitting, edge coloring, and orientations.
\newblock In {\em Proceedings of the {ACM-SIAM} Symposium on Discrete
  Algorithms (SODA)}, pages 2505--2523, 2017.

\bibitem[HKP98a]{HanckowiakKP98}
Micha{\l} Ha{\'n}{\'c}kowiak, Micha{\l} Karonski, and Alessandro Panconesi.
\newblock On the distributed complexity of computing maximal matchings.
\newblock In {\em Proceedings of the {ACM-SIAM} Symposium on Discrete
  Algorithms (SODA)}, pages 219--225, 1998.

\bibitem[HKP98b]{HanckowiakKaronskiPanconesi2001MM}
Micha{\l} Ha{\'n}{\'c}kowiak, Micha{\l} Karonski, and Alessandro Panconesi.
\newblock On the distributed complexity of computing maximal matchings.
\newblock In {\em Proceedings of the {ACM-SIAM} Symposium on Discrete
  Algorithms (SODA)}, pages 219--225, 1998.

\bibitem[HKP99]{hanckowiak1999faster}
Micha{\l} Ha{\'n}{\'c}kowiak, Micha{\l} Karo{\'n}ski, and Alessandro Panconesi.
\newblock A faster distributed algorithm for computing maximal matchings
  deterministically.
\newblock In {\em Proceedings of the {ACM} Symposium on Principles of
  Distributed Computing (PODC)}, pages 219--228, 1999.

\bibitem[HS12]{Hirvonen:2012}
Juho Hirvonen and Jukka Suomela.
\newblock Distributed maximal matching: Greedy is optimal.
\newblock In {\em Proceedings of the {ACM} Symposium on Principles of
  Distributed Computing (PODC)}, pages 165--174, 2012.

\bibitem[HSS16]{harris2016distributed}
David~G. Harris, Johannes Schneider, and Hsin-Hao Su.
\newblock Distributed (${\Delta}$+ 1)-coloring in sublogarithmic rounds.
\newblock In {\em Proceedings of the Symposium on Theory of Computing (STOC)},
  pages 465--478, 2016.

\bibitem[II86]{israeli1986fast}
Amos Israeli and Alon Itai.
\newblock A fast and simple randomized parallel algorithm for maximal matching.
\newblock {\em Information Processing Letters}, 22(2):77--80, 1986.

\bibitem[IS86]{israeli1986improved}
Amos Israeli and Yossi Shiloach.
\newblock An improved parallel algorithm for maximal matching.
\newblock {\em Information Processing Letters}, 22(2):57--60, 1986.

\bibitem[KMW06]{DBLP:conf/soda/KuhnMW06}
Fabian Kuhn, Thomas Moscibroda, and Roger Wattenhofer.
\newblock The price of being near-sighted.
\newblock In {\em Proceedings of the {ACM-SIAM} Symposium on Discrete
  Algorithms (SODA)}, pages 980--989, 2006.

\bibitem[KMW16]{DBLP:journals/jacm/KuhnMW16}
Fabian Kuhn, Thomas Moscibroda, and Roger Wattenhofer.
\newblock Local computation: Lower and upper bounds.
\newblock {\em J. {ACM}}, 63(2):17:1--17:44, 2016.

\bibitem[KY09]{koufogiannakis2009distributed}
Christos Koufogiannakis and Neal~E. Young.
\newblock Distributed fractional packing and maximum weighted b-matching via
  tail-recursive duality.
\newblock In {\em Proceedings of the International Symposium on Distributed
  Computing (DISC)}, pages 221--238, 2009.

\bibitem[Lin87]{linial1987LOCAL}
Nathan Linial.
\newblock Distributive graph algorithms - global solutions from local data.
\newblock In {\em Proceedings of the Symposium on Foundations of Computer
  Science (FOCS)}, pages 331--335, 1987.

\bibitem[Lin92]{linial1992locality}
Nathan Linial.
\newblock Locality in distributed graph algorithms.
\newblock {\em SIAM Journal on Computing}, 21(1):193--201, 1992.

\bibitem[LPSP08]{lotker2008improved}
Zvi Lotker, Boaz Patt-Shamir, and Seth Pettie.
\newblock Improved distributed approximate matching.
\newblock In {\em Proceedings of the {ACM} Symposium on Principles of
  Distributed Computing (PODC)}, pages 129--136, 2008.

\bibitem[LPSP15]{lotkerMatchingImproved}
Zvi Lotker, Boaz Patt-Shamir, and Seth Pettie.
\newblock Improved distributed approximate matching.
\newblock {\em Journal of the ACM (JACM)}, 62(5), 2015.

\bibitem[LPSR07]{Lotker:2007}
Zvi Lotker, Boaz Patt-Shamir, and Adi Rosen.
\newblock Distributed approximate matching.
\newblock In {\em Proceedings of the {ACM} Symposium on Principles of
  Distributed Computing (PODC)}, pages 167--174, 2007.

\bibitem[Lub86]{luby1986simple}
Michael Luby.
\newblock A simple parallel algorithm for the maximal independent set problem.
\newblock {\em SIAM journal on computing}, 15(4):1036--1053, 1986.

\bibitem[Mul92]{mulder1992julius}
Henry~Martyn Mulder.
\newblock Julius petersen's theory of regular graphs.
\newblock {\em Discrete mathematics}, 100(1-3):157--175, 1992.

\bibitem[PR01]{panconesirizzi2000}
Alessandro Panconesi and Romeo Rizzi.
\newblock Some simple distributed algorithms for sparse networks.
\newblock {\em Distributed computing}, 14(2):97--100, 2001.

\bibitem[PR07]{parnas2007approximating}
Michal Parnas and Dana Ron.
\newblock Approximating the minimum vertex cover in sublinear time and a
  connection to distributed algorithms.
\newblock {\em Theoretical Computer Science}, 381(1):183--196, 2007.

\bibitem[PS04]{pettie2004simpler}
Seth Pettie and Peter Sanders.
\newblock A simpler linear time 2/3- $\varepsilon$ approximation for maximum
  weight matching.
\newblock {\em Information Processing Letters}, 91(6):271--276, 2004.

\bibitem[PS10]{panconesi2010fast}
Alessandro Panconesi and Mauro Sozio.
\newblock Fast primal-dual distributed algorithms for scheduling and matching
  problems.
\newblock {\em Distributed Computing}, 22(4):269--283, 2010.

\bibitem[RTVX11]{Rubinfeld2011LCA}
Ronitt Rubinfeld, Gil Tamir, Shai Vardi, and Ning Xie.
\newblock Fast local computation algorithms.
\newblock In {\em Proceedings of the Symposium on Innovations in Computer
  Science (ICS)}, pages 223--238, 2011.

\bibitem[Suo10]{suomela2010EdgeDominatingSets}
Jukka Suomela.
\newblock Distributed algorithms for edge dominating sets.
\newblock In {\em Proceedings of the {ACM} Symposium on Principles of
  Distributed Computing (PODC)}, pages 365--374, 2010.

\bibitem[WW04]{wattenhofer2004distributed}
Mirjam Wattenhofer and Roger Wattenhofer.
\newblock Distributed weighted matching.
\newblock In {\em Proceedings of the International Symposium on Distributed
  Computing (DISC)}, pages 335--348, 2004.

\bibitem[YG80]{YannakakisGavril1980EdgeDominatingSets}
Mihalis Yannakakis and Fanica Gavril.
\newblock Edge dominating sets in graphs.
\newblock {\em SIAM Journal on Applied Mathematics}, 38(3):364--372, 1980.

\end{thebibliography}


\end{document}